%% file: main.tex
\documentclass[11pt]{article}

\usepackage[letterpaper,margin=1in]{geometry}
\usepackage[T1]{fontenc}
\usepackage[utf8]{inputenc}

\usepackage{amsmath,amssymb,amsfonts,amsthm,mathtools}
\usepackage{aliascnt}
\usepackage{microtype}
\usepackage{xcolor}
\usepackage{graphicx}
\usepackage{booktabs}   
\usepackage{array}
\usepackage{enumitem}
\usepackage{parskip}    

\usepackage{hyperref}
\hypersetup{
  colorlinks=true,
  linkcolor=blue,
  citecolor=blue,
  urlcolor=blue,
  filecolor=blue,
  anchorcolor=blue,
  pdftitle={The Price of Order in the Logarithmic Method},
  pdfauthor={Sichen Wang, Zhipeng Lu, and Jingbang Chen}
}

\theoremstyle{plain}
\newtheorem{theorem}{Theorem}[section]
\newaliascnt{proposition}{theorem}
\newtheorem{proposition}[proposition]{Proposition}
\aliascntresetthe{proposition}
\newaliascnt{lemma}{theorem}
\newtheorem{lemma}[lemma]{Lemma}
\aliascntresetthe{lemma}
\newaliascnt{corollary}{theorem}
\newtheorem{corollary}[corollary]{Corollary}
\aliascntresetthe{corollary}
\newaliascnt{conjecture}{theorem}

\aliascntresetthe{conjecture}

\theoremstyle{definition}
\newaliascnt{definition}{theorem}
\newtheorem{definition}[definition]{Definition}
\aliascntresetthe{definition}

\theoremstyle{remark}
\newaliascnt{remark}{theorem}
\newtheorem{remark}[remark]{Remark}
\aliascntresetthe{remark}
\newaliascnt{example}{theorem}

\aliascntresetthe{example}

\usepackage[capitalize,nameinlink]{cleveref}
\crefname{theorem}{Theorem}{Theorems}
\Crefname{theorem}{Theorem}{Theorems}
\crefname{proposition}{Proposition}{Propositions}
\Crefname{proposition}{Proposition}{Propositions}
\crefname{lemma}{Lemma}{Lemmas}
\Crefname{lemma}{Lemma}{Lemmas}
\crefname{corollary}{Corollary}{Corollaries}
\Crefname{corollary}{Corollary}{Corollaries}
\crefname{conjecture}{Conjecture}{Conjectures}
\Crefname{conjecture}{Conjecture}{Conjectures}
\crefname{definition}{Definition}{Definitions}
\Crefname{definition}{Definition}{Definitions}
\crefname{remark}{Remark}{Remarks}
\Crefname{remark}{Remark}{Remarks}
\crefname{example}{Example}{Examples}
\Crefname{example}{Example}{Examples}
\crefname{section}{Section}{Sections}
\Crefname{section}{Section}{Sections}
\crefname{appendix}{Appendix}{Appendices}
\Crefname{appendix}{Appendix}{Appendices}

\numberwithin{equation}{section}

\newcommand{\Rhat}{\widehat{R}}        
\newcommand{\What}{\widehat{W}}        
\newcommand{\Dtil}{\widetilde{D}}      

\title{The Price of Order in the Logarithmic Method}
\author{
  \begin{tabular}{@{}c@{\hspace{2.5em}}c@{\hspace{2.5em}}c@{}}
    Sichen Wang
      & Zhipeng Lu
      & Jingbang Chen\textsuperscript{*}\\[0.15em]
    {\footnotesize\texttt{wsc@smbu.edu.cn}}
      & {\footnotesize\texttt{zhipeng.lu@smbu.edu.cn}}
      & {\footnotesize\texttt{chenjb@cuhk.edu.cn}}\\[-0.05em]
    {\footnotesize Shenzhen MSU-BIT University}
      & {\footnotesize Shenzhen MSU-BIT University}
      & {\footnotesize CUHK-Shenzhen \& SLAI}
  \end{tabular}
}
\date{}

\begin{document}
\pagenumbering{arabic}

\maketitle
\begingroup
\renewcommand{\thefootnote}{\fnsymbol{footnote}}
\footnotetext[1]{Corresponding author.}
\endgroup

\begin{abstract}
The logarithmic method is a classical static-to-dynamic transformation: it stores
one dynamic ordered set as several immutable static components and rebuilds them
by merges.  The same component-and-merge discipline underlies write-optimized
ordered indexes, where cheap insertions must be reconciled with exact ordered
queries.  In this paper, we study the insertion-only version after $n$ insertions, over
abstract keys, in a strongly materialized merge-stack model with sequential
component merges and one forward scan of the live components per query.  We bound
the product between the total amount of data written during the $n$ insertions
and the worst-case amount of data read by a single query, known as the
\emph{write-read product}.  The optimal bounds are as follows:
\begin{itemize}[leftmargin=1.6em,itemsep=1pt,topsep=2pt]
  \item Membership and local certificates: $\Theta(n\log^2 n)$.
  \item Order and range queries with named keys or endpoints:
  $\Theta(n\log^3 n)$.
  \item Select: $\Theta(n^2)$.
\end{itemize}
Thus, the logarithmic method does not impose a universal dynamic overhead:
under materialized one-way access, the optimum depends on what information the
query reveals before the scan starts.  This pinpoints the access-model
obstruction behind the extra logarithm for exact order and range queries, and
the quadratic barrier for select.
\end{abstract}

\input{sections/intro}
\input{sections/model}
\input{sections/membership}
\input{sections/order}
\input{sections/select}
\input{sections/lifting}
\input{sections/accessmodel}

\appendix
\crefalias{section}{appendix}
\input{sections/nonmaterialized}

\section*{Acknowledgments}
The authors used OpenAI's ChatGPT 5.5 Pro to help organize and polish the exposition of this paper, based on drafts written by the authors. All mathematical statements, proofs, and references were checked, revised, and finalized by the authors. The authors assume responsibility for all content of the submission.

\pagebreak
\bibliographystyle{alpha}
\bibliography{references}

\end{document}

%% file: sections/intro.tex
\section{Introduction}\label{sec:intro}

A classical theme in dynamic data structures is to obtain dynamism by
rebuilding static structures in batches.  The logarithmic method of Bentley and
Saxe keeps several static structures of different ages and merges them as updates
arrive~\cite{bentleysaxe1980}; lower-bound and structural versions go back to
Mehlhorn and Overmars~\cite{mehlhorn1981,overmars1983}.  This is also a natural
abstraction of write-optimized ordered storage.  The LSM-tree was introduced to
make high-rate updates sequential and cheap~\cite{oneil1996}; sorted-run
architectures appear in large-scale storage systems and key-value stores such as
Bigtable, Cassandra, MyRocks/RocksDB, WiscKey, and
PebblesDB~\cite{bigtable2008,cassandra2010,myrocks2020,wisckey2017,pebblesdb2017},
and in cache-oblivious write-optimized search structures~\cite{cola2007}.

The broad area is active because write optimization creates a persistent tension
between update cost, point lookup cost, range-query cost, and space.  The systems
literature formulates this tension explicitly through the RUM tradeoff, optimizes
merge policies and filter allocation for point lookups, and designs range filters
or range-query mechanisms for LSM-style stores~\cite{athanassoulis2016,monkey2017,dostoevsky2018,growlsm2025,surf2018,rosetta2020,remix2021}.
From a theory point of view, however, these systems raise a more basic question:
what does the logarithmic method itself do to the complexity of ordered-search
queries?

At a fixed time, a logarithmic-method structure is not one search tree.  It is a
set of independently built sorted components.  A point lookup can often be
certified locally in each component, for instance by a hash table or filter.  An
exact ordered query is different: predecessor, rank, and exact range queries must
know where the query key, or the interval endpoints, fall inside the relevant
components.  Thus the issue is not only the number of components, but also what
kind of information the query must recover inside each component.

\noindent\textbf{Question.}
For each query family $Q$, what is the optimal product of amortized write cost
and worst-case read cost for a logarithmic-method ordered set?
Equivalently, if $E$ is the total amount of data written and $\Rhat_Q$ is the
worst read cost for $Q$, we study
\[
  P_Q(n)=\inf_{\mathcal A}\left(\frac{E_{\mathcal A}(n)}{n}\right)
    \Rhat_{\mathcal A,Q}(n),
\]
and throughout the paper we report the unnormalized product $E\Rhat_Q$.

This question is not answered by the closest existing lines of work.  Classical
static-to-dynamic transformations analyze the overhead of maintaining many static
structures, but do not separate query families that all have logarithmic static
search cost~\cite{bentleysaxe1980,mehlhorn1981,overmars1983,competitivedyn2021}.
LSM and write-optimized-storage work optimizes concrete tradeoffs for point and
range queries, but does not give a matching query-by-query lower-bound
classification under a single merge history~\cite{athanassoulis2016,monkey2017,dostoevsky2018,surf2018,rosetta2020,remix2021}.
Fractional cascading explains how repeated searches across related catalogs can
be shared~\cite{chazelleguibas1986,mehlhornnaher1990,afshani2021}, but assumes
that the cross-catalog links needed by the search are actually available.  Lower
bounds in cell-probe and external-memory models study different access regimes,
where memory cells or blocks can be probed by address~\cite{fredmansaks1989,brodalfagerberg2003,larsen2012,patrascu2011unify}.
Our focus is the restricted but natural merge-stack regime in which components
are written by sequential merges and later read once in chronological order.

The main obstruction is easy to see in a small example.  Suppose the current set
is stored in three sorted components $C_1,C_2,C_3$, scanned in that order.  For
membership of a key $x$, each component only has to certify whether $x$ occurs
there.  For predecessor or rank of $x$, the query must learn the local gap of $x$
inside each $C_i$.  A search in $C_1$ does not reveal the gap in $C_2$ or $C_3$
unless some cross-component search shortcut has already been materialized.  The
natural algorithmic objection is fractional cascading: repeated searches in many
sorted catalogs are precisely what fractional cascading was designed to remove
\cite{chazelleguibas1986}, and dynamic variants maintain such links under updates
\cite{mehlhornnaher1990,afshani2021}.

In a materialized one-way merge stack, however, the shortcut has the wrong time
direction.  The query scans older components before younger ones.  A bridge into a
younger component would have to be encountered in an older component, but the
younger component did not exist when the older one was written.  Rewriting the
older component later changes its birth time and moves it behind the younger one
in the scan.  Thus the components remain independent sorted catalogs for the
purpose of exact order localization.

We formalize this obstruction in a merge-stack model.  An insertion-only ordered
set is stored as immutable sorted \emph{components}.  Updates append singleton
components and rebuild suffixes by sequential merges.  Queries make one forward
pass over the live components, from old to young, and cannot return to a component
after passing it.  Our main theorem is for \emph{abstract keys}, which can be
compared and hashed but carry no word-RAM coordinate revealing their rank, and for
\emph{strongly materialized} components, which are self-contained encodings of
their own keys rather than directories for other components.  These assumptions
are the clean setting in which the missing fractional cascade is exposed.  The
select query gives the forward scan even less information: it supplies only a
rank $K$, not a target key value, and the threshold key may be in a component that
has already passed under the head.

\subsection{Our Results}

The main result is a tight, query-sensitive classification of the logarithmic
method in this model.

\begin{theorem}[Main Classification]\label{thm:main}
In the strongly materialized merge-stack model over abstract keys, with oblivious
merge schedules and bounded-error randomized query evaluation, the optimal
write-read product satisfies
\[
  \min E\Rhat_Q =
  \begin{cases}
    \Theta(n\log^2 n), & Q\text{ is locally certifiable in each component},\\[1mm]
    \Theta(n\log^3 n), & Q\text{ is a target-key order or exact range query},\\[1mm]
    \Theta(n^2),       & Q\text{ is select}.
  \end{cases}
\]
Here the first line includes membership, point-emptiness, and global
minimum/maximum.  The second line includes predecessor, successor, rank,
range-count, exact range-sum with unit weights, range-minimum, range-maximum, and
bounded-interval emptiness.
\end{theorem}

The minimum ranges over all algorithms in the model.  The upper bounds are
explicit merge schedules with forward-readable static indexes.  The lower bounds
are worst-case over insertion sequences and query instances.

\begin{center}
\small
\setlength{\tabcolsep}{4pt}
\renewcommand{\arraystretch}{1.18}
\begin{tabular}{@{}>{\raggedright\arraybackslash}p{0.22\textwidth}>{\raggedright\arraybackslash}p{0.30\textwidth}>{\raggedright\arraybackslash}p{0.26\textwidth}>{\raggedright\arraybackslash}p{0.15\textwidth}@{}}
\toprule
family & examples & information available before the scan & product \\
\midrule
local certificates & membership, point-emptiness, global min/max & a key or constant-size certificate checked component by component & $\Theta(n\log^2 n)$ \\
target-key order/range & predecessor, rank, range-count, exact range aggregates & key values or interval endpoints whose local gaps must be found & $\Theta(n\log^3 n)$ \\
selection & select & only a rank $K$, with no target key value & $\Theta(n^2)$ \\
\bottomrule
\end{tabular}
\end{center}

\Cref{thm:main} is not an answer-size separation.  Membership, rank, and select
all have $\Theta(\log n)$-bit answers on the hard instances, and rank and select
are inverse operations.  The difference is what the scan knows before it starts.
Membership names a key and equality can be checked locally.  Predecessor, rank,
and range queries name key values, but exactness requires locating those values
inside each component.  Select names no key at all; the answer value is unknown
while early components pass under the one-way scan.

The assumptions in \Cref{thm:main} are also close to the boundary of the
phenomenon.  Integer keys reduce the cost of searching within one component;
random access removes the no-backseek difficulty behind select; and
non-materialized cross-component directories can reintroduce fractional-cascading
information.  We treat these variants after the main proof, in \Cref{sec:access}
and \Cref{sec:nonmat}.  They are not part of the headline model, but they help
separate which restriction is responsible for which line of the classification.

\subsection{Overview}

The proof has a simple architecture.  For each query family we prove a read
barrier at a fixed state, and then combine it with an accounting lemma for all
suffix-merge histories.

For membership, the fixed-state barrier is the number of live components.  Let
$\What$ be the maximum number of live components, and let $\Dtil$ be the average
number of times an item is written, including its initial write.  The edit floor
and the merge-forest counting lemma give
\[
  E=\Omega(n\Dtil),
  \qquad
  \Dtil\What=\Omega(\log^2 n).
\]
Some membership query must inspect all live components, so $\Rhat=\Omega(\What)$.
This gives $E\Rhat=\Omega(n\log^2 n)$, matched by the balanced Bentley--Saxe
schedule with local hashes.

For target-key order and range queries, the read barrier is the summed static
search cost in the live components.  At a state with component sizes $m_i$, we
prove
\[
  \Rhat_{\mathrm{ord}}
    =\Omega\!\left(\sum_i \log(m_i+1)\right).
\]
The terms add because strong materialization forbids a read in one component from
revealing the private search position of the query key in another component.  The
merge-history side then proves the size-diversity law
\[
  E\cdot \max_t \sum_{C\text{ live at }t}\log(|C|+1)
    =\Omega(n\log^3 n).
\]
This is the only extra logarithm in the middle line of \Cref{thm:main}.  Once the
predecessor lower bound is established, direct-sum arguments transfer it to rank,
range-count, exact range-sum, range-minimum, range-maximum, and bounded-interval
emptiness.

For select, the query is target-free.  Consider a boundary in the scan with $P$
stored keys before it and $U$ stored keys after it.  We construct a rank query
whose answer is one of $\min\{P,U\}$ possible prefix keys, but the index of the
right key is determined only by suffix information.  Since the scan cannot go
back, it must read
\[
  \Omega(\min\{P,U\})
\]
words before crossing the boundary.  If the read budget is small, nearly all live
keys must repeatedly sit in one huge component; immutability then forces that
component to be rebuilt many times.  This yields $E=\Omega(n^2/\Rhat)$ and hence
$E\Rhat=\Omega(n^2)$.

\Cref{sec:model} defines the model, costs, query families, and key assumptions.
\Cref{sec:membership} proves the membership line.  \Cref{sec:order} proves the
no-cascade read floor, the size-diversity law, and the order/range line.
\Cref{sec:select} proves the cut lower bound and the quadratic select line.
\Cref{sec:lifting} gives a general static-to-dynamic lifting principle for
additive queries.  \Cref{sec:access} records access-model variants, and
\Cref{sec:nonmat} studies cascades beyond strong materialization.

\subsection{Related Works}
\label{sec:related}

The closest prior work explains pieces of the landscape: dynamization by merging,
write-optimized storage, fractional cascading, and lower bounds in richer memory
models.  The following table summarizes why none of these directions already
implies \Cref{thm:main}.

\begin{center}
\small
\setlength{\tabcolsep}{3pt}
\renewcommand{\arraystretch}{1.16}
\begin{tabular}{@{}>{\raggedright\arraybackslash}p{0.21\textwidth}>{\raggedright\arraybackslash}p{0.30\textwidth}>{\raggedright\arraybackslash}p{0.38\textwidth}@{}}
\toprule
line of work & what it explains & difference from this paper \\
\midrule
logarithmic method & how to maintain many rebuilt static structures & no classification separating query families with similar static costs \\
LSM and write-optimized storage & practical read/write/space tradeoffs for sorted runs & not a matching lower-bound theorem for membership, order/range, and select under one merge history \\
fractional cascading & how to share repeated searches across related catalogs & the useful bridge has the wrong chronological direction in a materialized one-way merge stack \\
cell-probe and external-memory lower bounds & tradeoffs under random access to cells or blocks & our restriction is not cell access but the order in which materialized information can be read and rewritten \\
\bottomrule
\end{tabular}
\end{center}

\paragraph{Static-to-Dynamic Transformations.}
Bentley and Saxe introduced the logarithmic method for decomposable searching
problems~\cite{bentleysaxe1980}.  Mehlhorn proved lower bounds for such
transformations~\cite{mehlhorn1981}, and Overmars developed a systematic theory
of dynamic structures built from static ones~\cite{overmars1983}.  Recent
competitive dynamization studies merge policies under nonuniform inputs and
read/write ratios~\cite{competitivedyn2021}.  These works explain the overhead of
maintaining many static structures.  They do not distinguish query types whose
static problems are all logarithmic; the extra logarithm in \Cref{thm:main} is
caused by exact order localization across materialized components.

\paragraph{Write-Optimized Storage.}
The LSM-tree was introduced to make high-rate updates sequential and cheap
\cite{oneil1996}.  LSM-style sorted runs appear in large-scale storage systems and
key-value stores~\cite{bigtable2008,cassandra2010,myrocks2020,wisckey2017,pebblesdb2017}.
The systems literature maps the read/update/storage tradeoff~\cite{athanassoulis2016},
optimizes merge policies and filter allocation for point lookups
\cite{monkey2017,dostoevsky2018,growlsm2025}, and builds succinct or probabilistic
range filters and range-query mechanisms~\cite{surf2018,rosetta2020,remix2021}.
This work motivates the access discipline, but it does not prove a tight
query-by-query product classification under the same merge histories.

\paragraph{Fractional Cascading.}
Fractional cascading removes repeated binary searches across related catalogs in
static structures~\cite{chazelleguibas1986}; dynamic fractional cascading studies
how to maintain such links under updates~\cite{mehlhornnaher1990,afshani2021}.
Our obstruction is orthogonal.  The issue is not only the cost of updating a
cascade, but whether the bridge can appear before the component it accelerates is
created.  In the strongly materialized one-way model it cannot; \Cref{sec:nonmat}
examines what changes when such cross-component data are allowed.

\paragraph{Lower Bounds and Restricted Access.}
External-memory dictionary lower bounds study update/query tradeoffs under random
block access~\cite{brodalfagerberg2003,dineknich2022}.  Cell-probe lower bounds
for dynamic data structures, including dynamic range counting, use a more
powerful random-access memory model~\cite{fredmansaks1989,larsen2012,patrascu2011unify,ko2025}.
Our bounds are orthogonal: component encodings are static and transparent, but the
read and rewrite order is restricted.  The lifting theorem in \Cref{sec:lifting}
connects our merge-stack bounds to static predecessor and range-searching lower
bounds~\cite{patrascuthorup2006,patrascuthorup2007,patrascu2007,chanwilkinson2013}.
The merge-history accounting is closest in spirit to list labeling and file
maintenance~\cite{ikr1981,bks2012,bcfc2022,nearlyoptimal2024}; the select lower
bound is related to selection under restricted access~\cite{munropaterson1980,fredericksonjohnson1984,kkzz2019,borst2023}.

%% file: sections/model.tex
\section{Preliminaries}\label{sec:model}

This section fixes the model used in the main classification.  The objects are
merge histories of immutable sorted components, the costs are total rewritten mass
and prefix-worst read cost, and the key assumptions isolate the absence of
cross-component order information.

\subsection{Merge Histories and Components}

A \emph{merge stack} maintains a sequence of immutable sorted components
$C_1,\dots,C_W$, oldest at the bottom.  Each component stores a sorted block of
key-value samples together with any static index built from that block alone.
Three operations maintain it:
\begin{itemize}
  \item \emph{insert}$(x,y)$ pushes a new singleton component $\{(x,y)\}$ on top;
  \item \emph{merge} pops a suffix $C_i,\dots,C_W$ and pushes a single component
    holding the union of their samples in sorted order;
  \item \emph{query} makes one forward pass over $C_1,C_2,\dots$ in stack order,
    reading component data; after an $O(1)$ header it may skip a component's body,
    but it never returns to an earlier component.
\end{itemize}
Inserts and merges interleave according to a schedule of the algorithm's choice.
The main lower bounds consider schedules that are oblivious to the key values;
this condition is part of strong materialization below.

\subsection{Cost Measures}

The \emph{write cost} $E$ is the total component data created or rewritten over
the execution.  The read cost of a query is the data it reads, and $\Rhat$ is its
prefix-worst value: the maximum, over every time $t$ and every query, of the cost
of answering that query against the state after step $t$.  Making $\Rhat$
prefix-worst keeps the read cost from being amortized away by periodic global
rebuilds, since every intermediate state must remain queryable.  For a randomized
evaluator the cost of a query is its expected number of charged reads, and
$\Rhat$ is again prefix-worst.  The complexity measure throughout is the product
$E\Rhat$.

A token is a machine word of $w=\Theta(\log n)$ bits, and keys and values are drawn
from universes of size $\mathrm{poly}(n)$, so a sample occupies $O(1)$ tokens and a
component of $m$ samples occupies $\Theta(m)$.  We use the following notation.

\begin{center}
\small
\setlength{\tabcolsep}{5pt}
\renewcommand{\arraystretch}{1.15}
\begin{tabular}{@{}ll@{}}
\toprule
symbol & meaning \\
\midrule
$E$ & total rewritten component data \\
$\Rhat_Q$ & prefix-worst read cost for query family $Q$ \\
$\What$ & maximum number of live components over the execution \\
$\Dtil$ & average item write depth, including the initial write \\
$L$ & $\log_2 n$ \\
$\lambda$ & $\log_2\log_2 n$ \\
\bottomrule
\end{tabular}
\end{center}

\subsection{Query Families}

Three query types drive the paper.  \emph{Membership} is the exact-match point
query: given $x$, return its value if $x$ is live and a default otherwise.  It is
locally certifiable because equality can be settled inside each component by a
hash or perfect dictionary.  \emph{Predecessor} is the basic target-key order
query: given $x$, return the value of the largest live key at most $x$.  Successor,
rank, and exact range queries -- count, sum, minimum, maximum, and emptiness over
an interval -- are target-key order or range queries of the same kind, since each
must locate a key value or endpoint in the live order.  Membership and these order
queries are \emph{decomposable}: an answer combines fixed per-component
contributions.  \emph{Select} is different.  It is given a rank $K$ and returns
the $K$-th smallest live key; no target key value is supplied to the scan.  Inserted
keys are distinct, and the lower bounds use distinct-key instances.

\subsection{Key Models and Materialization}

Keys come in two models, and the order separation depends on which one is used.
An \emph{abstract key} is the model for string, tuple, or opaque record keys: it is
hashable and comparable, but carries no integer structure.  We read it through an
oracle.  A key has an opaque equality label and an unknown position in the total
order, and neither the label nor its hash reveals that position.  A single read
returns at most one stored key handle together with $O(1)$ control words, and the
handle's rank among those already read is learned only by comparison.  The encoder
cannot manufacture an order-bearing separator it has not read, and an answer must
be a key or value decoded from the words read, not a pointer into an unread cell.
The $O(1)$ control words therefore cannot stand in for the order of a whole block.
An \emph{integer key} is an integer in a $\mathrm{poly}(n)$ universe, with the
word-RAM available: a van Emde Boas tree or y-fast trie locates the predecessor
inside a component in $O(\lambda)$ reads, independent of its size, and the order
gap narrows accordingly (\Cref{sec:access}).

The main theorems hold in executions realized by size-based log-structured
merging.  Ordinary materialization, which constrains only a component's contents,
is not enough: once the merge schedule may depend on the data, the partition into
components itself carries order information across them, and an unread block is no
longer free for an adversary to complete.  We use the following stronger
condition.

\begin{definition}[Strong Materialization]\label{def:strongmat}
An execution is \emph{strongly materialized} if it satisfies the following.
\begin{itemize}
  \item \emph{Canonical updates}: after each insertion the active components
    partition the inserted samples into contiguous blocks of consecutive
    insertion times, and each insertion performs at most one merge, replacing a
    suffix of the active components and the new singleton by a single component.
  \item \emph{Oblivious schedule}: the merged suffix at each step depends only on
    the step index and the current component sizes, not on the keys, values,
    query answers, or any coins.
  \item \emph{Exact layout}: a component of $m$ samples occupies a publicly fixed
    length $\ell(m)\in[c_0 m, c_1 m]$, so component boundaries, addresses, and the
    merge structure are functions of the size sequence alone.
  \item \emph{Local encoding}: a component's tokens depend only on its own
    samples, their insertion order and values, and comparisons among its own keys;
    inputs that agree on every component's internal order, labels, values, and
    insertion history, and differ only in how keys of distinct components
    interleave, induce identical component encodings.
  \item \emph{No external state}: no sample-dependent state persists outside the
    components, the query evaluator starts from a data-independent configuration
    with coins independent of the data, and every dependence on the data passes
    through reads of component tokens.
\end{itemize}
\end{definition}

This is what an LSM run already is, a self-contained sorted file; \Cref{sec:nonmat}
removes the condition and lets a young component share or copy order information
from older ones.

The per-component arguments use one consequence.  Call an execution
\emph{no-cascade} if a read in one component cannot shrink another component's set
of candidate answers.  Strong materialization supplies it: changing one block
leaves the partition, every other component's encoding, and the layout untouched,
so a read elsewhere says nothing about that block.  A non-materialized global key
order, searchable by random access, would violate this, the components then
admitting a joint search in $O(\log n)$ reads; random access alone, with the
components still materialized, does not.  This locality is what \Cref{sec:nonmat}
gives up.

The canonical-updates clause is without loss of generality: the merges performed
between two insertions act on nested suffixes of the stack and coalesce into their
final suffix merge, and collapsing them removes the intermediate writes and
queryable states, so neither $E$ nor $\Rhat$ grows.  Under exact layout the write
cost is structural rather than informational: building a component of $m$ samples
costs $\Theta(m)$ tokens whatever it holds, a fact the edit floor of
\Cref{sec:membership} reads off directly.

\subsection{Serialization and Forward Reads}

The lower bounds are cleaner in a low-level view of the same execution.  Serialize
the components in stack order into a single token string, each component a
contiguous block preceded by a length-prefix header.  Edits become tail operations:
an insert pushes one block at the tail, and a merge of a stack suffix pops the most
recently pushed blocks and pushes one new block.  Writing $\mathrm{LCP}$ for the
longest common prefix, the edit $s_{t-1}\to s_t$ costs
\[
  d_{\mathrm{LIFO}}(s_{t-1},s_t)
   =\bigl(|s_{t-1}|-|\mathrm{LCP}|\bigr)+\bigl(|s_t|-|\mathrm{LCP}|\bigr),
\]
the number of tokens popped and pushed at the tail, so
$E=\sum_t d_{\mathrm{LIFO}}(s_{t-1},s_t)$; we call such tail-only edits
\emph{LIFO}.  A query is run by a \emph{forward, no-backseek} head: it starts at
position $0$ and at each step either reads the current token, at unit cost,
advancing one place, or skips forward at no cost by an amount it computes from its
finite control, the query, and the tokens read so far.  It never moves backward,
working memory is free, and only reads are charged.

\begin{lemma}[Serialization Equivalence]\label{lem:serial}
For strongly materialized executions, the merge-stack model and the LIFO-token,
forward-no-backseek model agree up to constant factors in $(E,\Rhat)$.
\end{lemma}

\begin{proof}
Serialize as above.  An insert pushes a block of $\Theta(1)$ tokens at the tail,
and a merge pops exactly the most recently pushed blocks and pushes one new block,
so every edit is tail-only; the pushed token mass equals the rewritten component
data up to the matching pops, a constant factor, so $E_{\mathrm{token}}=\Theta(E)$.
The components occupy the string in stack order, which is birth order and the order
a forward head reads them, and skipping a body after its header is a legal
no-backseek pass, so $\Rhat_{\mathrm{token}}=\Theta(\Rhat)$.  Conversely, a
tail-only token execution whose pushed blocks are local encodings of the unions of
the popped blocks splits into contiguous blocks in birth order, since the tokens
pushed at one step are consecutive and share a birth time, and is a merge-stack
execution with the same merge structure.
\end{proof}

One may therefore keep the merge-stack picture throughout and read the lower-bound
proofs as statements about the token string.

\subsection{Baseline Upper Bounds}

The serial Bentley--Saxe transform~\cite{bentleysaxe1980} supplies the upper
bounds the lower bounds will match.

\begin{proposition}[Bentley--Saxe Baseline]\label{prop:baseline}
A single schedule has $E=O(n\log n)$ with at most $\lceil\log_2(n+1)\rceil$
active components at all times, and answers membership in $\Rhat=O(\log n)$,
abstract-key predecessor in $\Rhat=O(\log^2 n)$, and integer predecessor in
$\Rhat=O(\lambda\log n)$. The corresponding products are
\[
  E\Rhat=O(n\log^2 n),\qquad O(n\log^3 n),\qquad O(n\log^2 n\,\lambda).
\]
\end{proposition}

\begin{proof}
Keep the components whose sizes are the set bits of the insert count, merging
equal-sized ones on each insert.  A component of size $2^k$ is built at most
$n/2^k$ times, each at $\Theta(2^k)$ tokens, so $E=O(n\log n)$, and at most
$\lceil\log_2(n+1)\rceil$ components are live at any time.  Each component carries
a forward-readable static structure, length-prefixed so the head skips its body
after one header read.  A static perfect hash answers membership in $O(1)$ reads
per component, even against an adversarial query, so $\Rhat=O(\log n)$.  For
abstract-key predecessor, a balanced search tree laid out in search-round order,
all depth-one nodes, then all depth-two, and so on, is traversed root to leaf as a
forward scan of $O(\log m_k)$ comparison reads in component $k$, summing to
$O(\log^2 n)$ over the active components.  For integer predecessor, a y-fast trie
answers in $O(\lambda)$ reads per component under the same round-by-round layout,
giving $\Rhat=O(\lambda\log n)$.
\end{proof}

The membership and predecessor products are the first the lower bounds match
(\Cref{sec:membership,sec:order}); select needs a different schedule and is
polynomially larger (\Cref{sec:select}).  How the predecessor bound responds when
these conditions are relaxed, and why the relaxations are not interchangeable, is
the subject of \Cref{sec:access}.

%% file: sections/membership.tex
\section{Membership}\label{sec:membership}

The Bentley--Saxe schedule of \Cref{sec:model} answers membership with product
$n\log^2 n$. We show this is best possible, and in doing so establish the three
structural facts the later lower bounds reuse: a level decomposition that the
tail-only discipline forces, an edit floor on the write cost, and a counting
bound that trades depth against frontier. Together they expose the mechanism the
paper turns on, that $E\Rhat$ factors into a write budget and a sum of per-level
read charges. Membership is the case where a hash holds every charge at $O(1)$;
order and rank raise it.

\subsection{Merge Forests}

Fix an execution. A live token in $s_t$ has a \emph{birth step}, the step at
which it was last pushed, and tokens with a common birth step form a
\emph{level}; by \Cref{lem:serial} a level is just a component of \Cref{sec:model}
in token-string form. Write $W(t)$ for the number of live levels at time $t$.

\begin{lemma}[Level Decomposition]\label{lem:leveldecomp}
At every time $t$ the live levels are contiguous and ordered by birth,
$s_t=L_{t_1}\cdots L_{t_W}$ with $t_1<\cdots<t_W$, and the edit $s_{t-1}\to s_t$
pops a suffix of them and pushes one new top level.
\end{lemma}

\begin{proof}
A canonical update pops a suffix of the active components and pushes one merged
component (\Cref{sec:model}). The tokens pushed at one step are consecutive at
the tail and share a birth step, so each level is contiguous and older levels lie
earlier; an edit pops a suffix of whole levels and makes the pushed tokens the
new top level.
\end{proof}

The decomposition makes the execution a \emph{merge forest}: an internal node for
each step that merges, its children the levels popped there and the singleton
inserted then, and a leaf for each insert. Leaf $j$ takes part in $d_j$ merges,
one for each of its internal ancestors. The rewritten sample mass is therefore
\[
  M:=\sum_j d_j=\sum_P|P|,
\]
$|P|$ being the number of leaves below the internal node $P$. The forest is fixed by the
schedule alone, whatever the tokens mean, and the lower bounds read off it.

The hard instances carry their entropy in the values: each inserted value is
independent and uniform over a $\mathrm{poly}(n)$ universe, so it is
$\Theta(\log n)$ bits, the content of one token.

\begin{lemma}[Incompressibility]\label{lem:inc}
On the i.i.d.-value instance, a materialized level that is the sole source for $m$
of its samples occupies $\Omega(m)$ tokens.
\end{lemma}

\begin{proof}
The state must determine those $m$ values, whose joint entropy is
$\Theta(m\log n)$ bits, or $\Theta(m)$ tokens. Materialization places them in this
level alone, so it occupies $\Omega(m)$ tokens.
\end{proof}

\begin{corollary}[Edit Floor]\label{cor:editfloor}
With \emph{augmented depth} $\Dtil:=(M+n)/n$, the write cost is
$E=\Omega(n\Dtil)$.
\end{corollary}

\begin{proof}
Under exact layout a level of $m$ samples costs $\Theta(m)$ tokens, so
$E=\Theta(n+M)=\Theta(n\Dtil)$. With no layout fixed, \Cref{lem:inc} applies to
each level when it is built, the sole source for its samples at that moment, and
summing over the levels built, of total size $n+M$, gives $E\ge c(n+M)=cn\Dtil$
for a constant $c\in(0,1]$.
\end{proof}

The third fact ties depth to frontier. Its counting argument is the
online-labeling lemma of Bul\'anek, Kouck\'y, and Saks~\cite{bks2012}, recast in
merge-forest terms. Write $\What:=\max_t W(t)$ for the \emph{frontier width}.

\begin{lemma}[Frontier Width]\label{lem:fw}
Every merge forest on $n$ leaves with augmented depth $\Dtil$ and frontier width
$\What$ has $\Dtil\What=\Omega(\log^2 n)$.
\end{lemma}

\begin{proof}
With a dummy root above the final frontier, a leaf in $d_j$ merges has augmented
depth $d_j+1$, of mean $\Dtil$. Let $N(h,W)$ be the most leaves an ordered suffix-merge
tree of height at most $h$ can have while its live frontier never exceeds $W$. Its
root pops a suffix of $r\le W$ children, each a completed subtree that holds a
frontier slot until the merge, together with the singleton inserted then, a
height-$0$ child holding no slot. With $N(0,W)=1$,
\[
  N(h,W)\ \le\ 1+\sum_{q=1}^{W}N(h-1,q)\ \le\ \sum_{q=0}^{W}\binom{h-1+q}{q}
   \ =\ \binom{h+W}{W},
\]
the singleton being the $q=0$ term and the last step the hockey-stick identity.
For $a\le b$, $\log_2\binom{a+b}{a}\le a\log_2\frac{e(a+b)}{a}=O(\sqrt{ab})$ using
$\log_2 u\le 2\sqrt u$. By Markov at least $n/2$ leaves have augmented depth at
most $2\Dtil$; keeping them and suppressing empty internal nodes leaves an ordered
tree with at least $n/2$ leaves, height at most $\lceil 2\Dtil\rceil$, and
frontier at most $\What$, so
\[
  \log_2\tfrac n2\ \le\ \log_2\binom{\lceil 2\Dtil\rceil+\What}{\What}
   \ =\ O\!\bigl(\sqrt{\Dtil\What}\bigr),
\]
that is $\Dtil\What=\Omega(\log^2 n)$.
\end{proof}

\subsection{The Membership Bound}

The read side needs one fact: a forward head must visit every live level.

\begin{lemma}[Sequential Access]\label{lem:seqaccess}
For an oblivious schedule, some strongly materialized realization forces
$\Rhat\ge\What$ on a fresh membership query against a deterministic evaluator, and
$\Rhat\ge\What/3$ against a bounded-error one.
\end{lemma}

\begin{proof}
At a time of maximum frontier, take the realization in which every live level
straddles a fresh key $x$ matching none of its samples. Each live level must
be decided on its own, that $x$ is absent there, and under strong materialization
no other level constrains it, so a level left unread can be filled with $x$ by the
adversary, flipping the answer. A deterministic evaluator therefore reads all
$\What$ levels, so $\Rhat\ge\What$. For a bounded-error one, couple this input
with the one in which only level $i$ gains $x$: the transcripts agree until level
$i$ is read, so skipping it makes the evaluator answer both alike and err on one,
and the two error probabilities sum to at least the probability of skipping level
$i$. Bounded error keeps that sum below $\tfrac23$, so level $i$ is read with
probability at least $\tfrac13$ and $\Rhat\ge\What/3$. A length header does not
escape the charge, that read being the one counted.
\end{proof}

\begin{theorem}[Tight Membership Bound]\label{thm:membership}
Strongly materialized membership has $E\Rhat=\Theta(n\log^2 n)$, attained at
$(E,\Rhat)=(\Theta(n\log n),\Theta(\log n))$ by hashing on the Bentley--Saxe
schedule.
\end{theorem}

\begin{proof}
The upper bound is \Cref{prop:baseline}. For the lower bound, \Cref{lem:seqaccess}
gives $\Rhat\ge\What$, \Cref{cor:editfloor} gives $E=\Omega(n\Dtil)$, and
\Cref{lem:fw} gives $\Dtil\What=\Omega(\log^2 n)$, so
$E\Rhat=\Omega(n\Dtil\What)=\Omega(n\log^2 n)$.
\end{proof}

Nothing in the lower bound used what membership computes, only that every live
level must be decided, so $n\log^2 n$ is a floor for the whole model, reached by
membership. It is robust to randomization through \Cref{lem:seqaccess}, and it
constrains the tradeoff and not just the product: $\Rhat=O(\log n)$ forces
$\What=O(\log n)$, hence $\Dtil=\Omega(\log n)$ and $E=\Omega(n\log n)$. Membership
is the hashable extreme, each level settled by one probe. The rest of the paper
keeps this write budget and asks whether the read side can stay at $O(1)$ per
level; once the query needs order, it cannot.

%% file: sections/order.tex
\section{Order Queries}\label{sec:order}

This section proves the middle line of \Cref{thm:main}.  The proof has three
steps.  First, strong materialization and one-way access rule out the
cross-component shortcut that fractional cascading would need.  Second, once the
components cannot share searches, a fixed state with component sizes $m_i$ forces
$\Omega(\sum_i\log(m_i+1))$ reads for a hard predecessor query.  Third, every
cheap suffix-merge history exposes some state where this summed log-size is
large, giving the product lower bound $\Omega(n\log^3 n)$.  The final subsection
transfers predecessor to rank and exact range queries.

\subsection{No Cascading Across Components}

Membership spends $O(1)$ per level because a hash names where the answer sits.
Order has no such shortcut and cannot build one by linking levels. Call a
\emph{forward bridge} from a level $L_i$ to a younger level $L_j$ any data stored
in $L_i$ that lets the head, once it has read $L_i$, reach the query's predecessor
in $L_j$ within $O(1)$ further reads.

\begin{lemma}[Causal-Bridge Obstruction]\label{lem:causalbridge}
Under forward, no-backseek access with LIFO edits, no forward bridge can be
consulted when it would help.
\end{lemma}

\begin{proof}
Levels are read oldest first (\Cref{lem:leveldecomp}). A bridge that speeds up a
younger level $L_j$ must be read before $L_j$, hence stored in an older level
$L_i$; but its content depends on $L_j$, fixed only when $L_j$ is built, so its
carrier must be younger than $L_j$. The two demands collide: when $L_i$ was built
$L_j$ did not yet exist, and on adversarial independent keys no guess survives.
Rebuilding $L_i$ after $L_j$ only pushes its birth past $L_j$, moving it behind in
read order, and a younger level pointing into an older one is reached only once
the older target has been passed.
\end{proof}

This is the obstruction behind Afshani's dynamic fractional-cascading lower
bound~\cite{afshani2021}, sharpened from a tax to an impossibility: the bridge a
cascade needs would have to be written before its target exists. The other route,
a young level carrying a directory into the older ones, is closed not by access
but by strong materialization, each component encoding only its own samples
(\Cref{sec:model}); \Cref{sec:nonmat} drops that assumption. What remains is the
cost of searching one level alone.

Within a level a hash still names one specific key, while the predecessor is
another, fixed by the level's order, that the hash of $x$ does not reach. Locating
it is a comparison search, and on a level of $m$ keys an adversarial query forces
$\Omega(\log m)$ reads, made precise against a hard gap distribution in
\Cref{lem:zerogap}.

\subsection{Per-State Search Lower Bound}

Fix a time and write $\Phi(t)=\sum_{\text{active }L}\log_2|L|$ for the summed
log-size of the live levels. Since the per-level searches do not combine
(\Cref{lem:causalbridge}), the read cost is their sum. This sum is not forced at
every state: with the live levels in disjoint key ranges a single header settles
all but one. What does hold is that each size skeleton, the live levels' sizes alone, admits
an adversarial key assignment making the sum necessary, which under an oblivious
schedule is all the lower bound needs.

\begin{lemma}[Zero-Side Ordered Gap]\label{lem:zerogap}
Fix stored keys $k_1<\cdots<k_m$ and handles $b<x$. For $0\le j\le m$ let $Z_j$
place $k_1,\dots,k_j<b<x<k_{j+1},\dots,k_m$, and for $0\le j<m$ let $H_j$ place
$k_1,\dots,k_j<b<k_{j+1}<x<k_{j+2},\dots,k_m$, so the predecessor of $x$ is $b$ on
$Z_j$ and $k_{j+1}$ on $H_j$. Any randomized comparison tree reads $\Omega(\log(m+1))$ keys in
expectation on $Z_J$, $J$ uniform, if it uses the known internal order of the
$k_r$, learns the order of an unread key only by reading it, and errs with
probability at most $\tfrac13$ over $\{Z_j,H_j\}$.
\end{lemma}

\begin{proof}
Take a deterministic tree first. A read key has three relevant outcomes,
$k_r<b$, $b<k_r<x$, and $x<k_r$. Let $E_0$ count the zero instances not answered
$b$, $E_1$ the hot instances answered $b$, and $\Lambda$ the distinct $b$-leaves
reached by correctly answered zero instances. The zero instances reaching one leaf
form an interval: if $Z_a$ and $Z_c$ ($a<c$) share a $b$-leaf then no $k_r$ with
$a<r\le c$ was read, being above $x$ in $Z_a$ and below $b$ in $Z_c$, so every
$H_h$ with $a\le h<c$ follows the same transcript and is wrongly answered $b$. A
leaf covering $z$ zero instances thus forces $z-1$ hot errors, whence
$\Lambda\ge(m+1)-E_0-E_1$. Average error at most $\tfrac5{12}$ over the uniform
mixture of the $2m+1$ instances gives $E_0+E_1\le\tfrac5{12}(2m+1)$ and
$\Lambda\ge(m+1)/6$. With one correct zero instance per leaf, at depths
$d_1,\dots,d_\Lambda$, the ternary Kraft inequality $\sum_\ell 3^{-d_\ell}\le1$
forces mean depth $\ge\log_3\Lambda$, so
\[
  \mathbb E_J[\text{reads on }Z_J]\ \ge\ \tfrac{\Lambda}{m+1}\log_3\Lambda\ =\ \Omega(\log(m+1)).
\]
For a randomized tree the coin-fixed error averages at most $\tfrac13$, so by
Markov a $\tfrac15$ fraction of the fixings have error at most $\tfrac5{12}$, and
averaging over them keeps $\Omega(\log(m+1))$.
\end{proof}

\begin{lemma}[Anchor-and-Zones Floor]\label{lem:anchorzones}
Fix a public size vector $(m_1,\dots,m_W)$ of an oblivious schedule and an
evaluator of error at most $\tfrac13$. Some strongly materialized state with these
sizes admits a query $x$ on which the evaluator has
$\Rhat\ge c\sum_{i=1}^W\log_2(m_i+1)$, for an absolute $c>0$.
\end{lemma}

\begin{proof}
We average over a distribution of states, so some realization meets the
expectation, and obliviousness makes it carry the prescribed sizes. Take level $1$
as anchor: fix its internal order, draw the query gap uniformly among the $m_1$
gaps with a stored lower endpoint, and put $x$ there with $b$ that endpoint.
Locating $b$ is a randomized predecessor search among equally likely gaps, one
ternary split per read, so a bounded-error evaluator spends $\Omega(\log(m_1+1))$
reads, and an $O(1)$ header holds only constantly many handles and cannot help.
For $i\ge2$ fix three private zones, a low slab below every anchor key, a high
slab above $x$, and a hot subzone inside $(b,x)$, all ordered publicly by level
index. Draw $J_i$ uniform in $\{0,\dots,m_i\}$; the zero placement puts
$k_{i,1},\dots,k_{i,J_i}$ in the low slab and the rest in the high slab, and the
hot alternative, when $J_i<m_i$, moves $k_{i,J_i+1}$ into the hot subzone. Expose
all outside information for free, $b$ included. Reading $k_{i,r}$ then reveals only
its zone, comparisons within the level giving the known index order and
comparisons across levels only the public zone order, so nothing finer about
$J_i$ is learned without reading level $i$. The task is \Cref{lem:zerogap} with
$m=m_i$ and inherits its error guarantee, costing $\Omega(\log(m_i+1))$. Summing
over the anchor and the levels gives the bound.
\end{proof}

The floor is existential over the states realizing a skeleton, not a property of
every state; specialized to one read per level it is the randomized membership
floor $\Rhat\ge\What/3$ of \Cref{thm:membership}. On the hard realization
$\Rhat\ge c\sum_k\log_2(m_k+1)\ge c\max\{\What,\Phi(t)\}$, and the separation turns
on whether a schedule can hold $\Phi$ down throughout while editing for only
$O(n\log n)$. It cannot.

\subsection{Merge-History Size Diversity}

A logarithm above membership needs $\Phi(t)=\Omega(\log^2 n)$ at some moment, and
no schedule holds $\Phi$ at the membership scale $O(\log n)$ throughout. The
size-diversity law prices the time-averaged potential from below by the edit
budget. It is an entropy-prefix argument on the merge forest, budget-free, the
budget entering only through the augmented depth $\Dtil$.

Add a dummy root above the final frontier, with $S_{v_0}=n$. For an internal node
$v$ write $S_v$ for the leaves below it, $q_v=S_v/n$, and $z_v=\log_2 S_v$; its
children in birth order have sizes $s_1,\dots,s_r$, weights $p_i=s_i/S_v$, entropy
$h_v=\sum_i p_i\log_2(1/p_i)$, and \emph{ordered prefix load}
$a_v=\sum_j p_j\sum_{i<j}\log_2 s_i$. A merge coalesces an active suffix, so
$r\le\What+1$.

\begin{lemma}[Four Identities]\label{lem:fourid}
For every merge forest, exactly and with no budget hypothesis,
\[
  \mathrm{(I)}\ \tfrac1n\!\int_0^n\!\Phi=\sum_v q_v a_v,\quad
  \mathrm{(II)}\ \sum_v q_v h_v=L,\quad
  \mathrm{(III)}\ \sum_v q_v h_v z_v\ge\tfrac12 L^2,\quad
  \mathrm{(IV)}\ \sum_v q_v=\Dtil.
\]
\end{lemma}

\begin{proof}
For (I), read the integral as a sum over the intervals between insertions; it is
the lifetime identity $\int_0^n\Phi=\sum_L\tau_L\log_2|L|$, where $\tau_L$, the
inserts during $L$'s life, equals the mass of $L$'s younger siblings at its
parent, since a LIFO merge carries off exactly the levels younger than $L$. As
$S_v a_v=\sum_i(\sum_{j>i}s_j)\log_2 s_i=\sum_i\tau_{v_i}\log_2 s_i$, summing over
$v$ and dividing by $n$ gives (I). For (II) and (III), sample a leaf uniformly and
let $Z_0=L>Z_1>\cdots\ge0$ be the log-sizes along its root-to-leaf path, with
$X_k=Z_k-Z_{k+1}$ and $\sum_k X_k=L$. At a node $v$ on the path
$\mathbb E[X_k\mid v]=h_v$ and $\Pr[v\text{ on path}]=q_v$, giving (II); and
$\sum_k Z_k X_k=\tfrac12(\sum_k X_k)^2+\tfrac12\sum_k X_k^2\ge\tfrac12 L^2$ with
$\mathbb E[Z_k X_k\mid v]=z_v h_v$, giving (III). For (IV),
$\sum_v q_v=\tfrac1n\sum_{\text{leaves}}(d_j+1)=\Dtil$, since each leaf has
$d_j+1$ internal ancestors.
\end{proof}

Identity (III) is the engine: entropy weighted by current log-size unavoidably
reaches $\Omega(L^2)$, exactly. Turning it into a bound on (I) is a local
computation.

\begin{lemma}[Local Prefix Bound]\label{lem:localprefix}
Fix $\varepsilon\in(0,\tfrac14)$, $\gamma=\varepsilon^2$, $c_0=L^{-7}$, and write
$\mu=G^{-1}$ for $G(m)=(m+1)\log_2(m+1)-m\log_2 m$. Uniformly over internal nodes
with $z_v\ge\varepsilon L$, $h_v\ge c_0$, and $r\le n^\gamma+1$,
\[
  a_v\ \ge\ (1-\kappa)\,z_v\,\mu(h_v),\qquad \kappa\le\varepsilon+o(1).
\]
\end{lemma}

\begin{proof}
Put $\theta=(\log_2 r+9\log_2 L)/(\varepsilon L)$, so $\theta=\varepsilon+o(1)$
under the hypotheses. Call a child \emph{good} if $p_i\ge2^{-\theta z_v}$, so
$\log_2 s_i=z_v+\log_2 p_i\ge(1-\theta)z_v$. A bad child has
$p_i<2^{-\theta\varepsilon L}=(rL^9)^{-1}$, so the bad mass is at most $L^{-9}$,
and renormalizing onto the good children shifts the entropy by $O(L^{-8})$ and
hence $\mu$ by a factor $1-o(1)$: since $(\ln\mu)'=O(1+1/h)$ and $h_v\ge c_0$, the
shift in $\ln\mu$ is $O((1+L^7)L^{-8})=O(L^{-1})$. The maximum-entropy bound, that
a law on $\{0,1,2,\dots\}$ of entropy $h$ has mean at least $\mu(h)$, applied to a
child drawn by $p_j$, gives at least $(1-o(1))\mu(h_v)$ expected earlier good
children; each adds $\ge(1-\theta)z_v$ to $a_v$, so
$a_v\ge(1-\theta-o(1))z_v\mu(h_v)$.
\end{proof}

\begin{theorem}[Size-Diversity Law]\label{thm:scalediv}
Every LIFO merge forest with $\What\le n^\gamma$ and $\Dtil\le L^6$ satisfies
\[
  \tfrac1n\!\int_0^n\!\Phi(t)\,dt\ \ge\ (1-\varepsilon-o(1))\,\Dtil\,L\,
   \mu\!\Bigl(\tfrac{\beta L}{\Dtil}\Bigr),\qquad \beta=\tfrac12-\varepsilon-o(1),
\]
so $\max_t\Phi(t)$ is at least the same value.
\end{theorem}

\begin{proof}
Discard from (III) the nodes with $z_v<\varepsilon L$, contributing at most
$\varepsilon L\sum_v q_v h_v=\varepsilon L^2$ by (II), and those with $h_v<c_0$,
contributing at most $c_0 L\Dtil=o(L^2)$ by (IV). The survivors have
$z_v\ge\varepsilon L$, $h_v\ge c_0$, $r\le\What+1\le n^\gamma+1$, and carry
$\sum^{\mathrm s}q_v h_v z_v\ge\beta L^2$. Write $Y_v=q_v h_v z_v/L^2$, so
$m_Y:=\sum^{\mathrm s}Y_v\ge\beta$ and
$\sum^{\mathrm s}Y_v/h_v=L^{-2}\sum^{\mathrm s}q_v z_v\le\Dtil/L$ by (IV). Set
$\varphi(u)=u\,\mu(1/u)$, convex and decreasing, and note
$q_v z_v\mu(h_v)=L^2 Y_v\varphi(1/h_v)$. Identity (I) and the local prefix bound
give
\[
  \tfrac1{nL^2}\!\int\!\Phi\ \ge\ (1-\kappa)\,m_Y\,\mathbb E[\varphi(1/h)],
\]
the expectation under $Y/m_Y$. Jensen and then the monotonicity of $\varphi$ at
$\mathbb E[1/h]\le\Dtil/(Lm_Y)$ give
\[
  m_Y\,\mathbb E[\varphi(1/h)]\ \ge\ m_Y\,\varphi\!\Bigl(\tfrac{\Dtil}{Lm_Y}\Bigr)
   \ =\ \tfrac{\Dtil}{L}\,\mu\!\Bigl(\tfrac{m_Y L}{\Dtil}\Bigr)
   \ \ge\ \tfrac{\Dtil}{L}\,\mu\!\Bigl(\tfrac{\beta L}{\Dtil}\Bigr),
\]
using $m_Y\ge\beta$ and $\mu$ increasing. Multiply by $L^2$.
\end{proof}

The frontier cap is necessary for the law, since the never-merge schedule has
$\Phi\equiv0$, but it does not constrain the predecessor bound: the assembly below
routes wide-frontier executions through the read floors instead. The cap aside,
the law is the right general form rather than a relaxation, asymptotically tight
along the lazy base-$q$ counter, so the uniform reading $\max_t\Phi\ge c\log^2 n$
for all schedules is false, that family sending $\max_t\Phi/\log^2 n\to0$ as $q$
grows. Since $\Phi$ depends only on the forest, the bound holds for every
realization of a randomized algorithm.

\begin{lemma}[Calculus Floor]\label{lem:calcfloor}
$\min_{h>0}\mu(h)/h^2=\tfrac14$, attained only at $h=2$; equivalently
$G(m)\le2\sqrt m$ with equality iff $m=1$. Hence
$\Dtil^2\,\mu(\beta L/\Dtil)\ge(\beta^2/4)L^2$.
\end{lemma}

\begin{proof}
Substituting $m=\mu(h)$, the claim $\mu(h)\ge h^2/4$ reads $G(m)\le2\sqrt m$. Let
$F(m)=2\sqrt m-G(m)$, so $F(0^+)=F(1)=0$ and $F'(m)=m^{-1/2}-\log_2(1+1/m)$, of the
sign of $1-u(m)$ for $u(m)=\sqrt m\,\log_2(1+1/m)$. Here $u$ is unimodal, with
$u'(m)=0$ at a single point since $(1+1/m)^{m+1}$ decreases strictly to $e$, and
with $u\to0$ at both ends and $u(1)=1$ it exceeds $1$ on exactly one interval
$(m_1,1)$. So $F$ rises on $(0,m_1)$, returns to $F(1)=0$, then rises, giving
$F\ge0$ with equality only at $m\in\{0,1\}$. The lone interior contact $m=1$ gives
$h=G(1)=2$ and $\mu(2)=1$, so the minimum is $\tfrac14$; substituting
$h=\beta L/\Dtil$ gives the last inequality.
\end{proof}

\subsection{The Predecessor Bound}

\begin{theorem}[Tight Predecessor Bound]\label{thm:predecessor}
Every strongly materialized algorithm with a bounded-error evaluator, under no
hypothesis on the frontier $\What$ or the budget $M$, has predecessor
$E\Rhat=\Omega(n\log^3 n)$, which binary Bentley--Saxe attains. Predecessor is
$\Theta(n\log^3 n)$, one clean logarithm above membership.
\end{theorem}

\begin{proof}
The upper bound is \Cref{prop:baseline}. Fix the algorithm; its oblivious schedule
fixes the size trajectory and $E$ independently of the keys. Fix
$\varepsilon=10^{-2}$, $\gamma=\varepsilon^2$, $\Dtil=(M+n)/n$, and split on the
schedule.

\emph{(i) $\What>n^\gamma$.} \Cref{lem:anchorzones} at a maximum-frontier skeleton
gives $\Rhat=\Omega(\What)>\Omega(n^\gamma)$, and $E\ge n$, so
$E\Rhat=\Omega(n^{1+\gamma})\gg n\log^3 n$.

\emph{(ii) $\What\le n^\gamma$, $\Dtil>L^6$.} The edit floor gives
$E=\Omega(n\Dtil)>\Omega(nL^6)$ and \Cref{lem:anchorzones} at the final skeleton,
where $\sum_k\log_2(m_k+1)\ge\log_2(n+1)$, gives $\Rhat=\Omega(L)$, so
$E\Rhat=\Omega(nL^7)\gg n\log^3 n$.

\emph{(iii) $\What\le n^\gamma$, $\Dtil\le L^6$.} The size-diversity law gives a
time $t^\star$ with $\Phi(t^\star)\ge(1-\varepsilon-o(1))\Dtil L\,\mu(\beta L/\Dtil)$;
\Cref{lem:anchorzones} there gives $\Rhat=\Omega(\Phi(t^\star))$, the edit floor
gives $E=\Omega(n\Dtil)$, and the calculus floor closes it:
\[
  E\Rhat\ \ge\ \Omega\bigl(nL\cdot\Dtil^2\mu(\beta L/\Dtil)\bigr)
   \ \ge\ \Omega\bigl(\tfrac{\beta^2}4\,nL^3\bigr)=\Omega(nL^3).
\]
The cases are exhaustive. Obliviousness makes the size trajectory the same on
every input, so the skeleton each case invokes occurs along the execution, and
\Cref{lem:anchorzones} exhibits a hard realization there; the product is the
algorithm's worst case over inputs, in expectation over its coins.
\end{proof}

The bound $\Rhat\ge\What$ alone yields only $n\log^2 n$, one logarithm short;
the extra logarithm needs $\Rhat=\Omega(\Phi)$ at a high-$\Phi$ time, which the
size-diversity law supplies. Predecessor pays for size diversity, not merely
frontier width: membership is priced by $\What$, order by $\Phi$.

\subsection{Extensions to Rank and Range Queries}

Predecessor is one order query among many. Whether a decomposable query pays the
extra logarithm is read off the combiner, from how much of a single block's
contribution can reach the global answer, and for natural exact queries this takes
one of two sizes. It is $O(1)$ when an $O(1)$-read witness settles a block's
contribution, as for membership, point-emptiness, and the global minimum or
maximum. It is $\Theta(\log m)$ when the contribution turns on $q$'s position
among a block's $m$ sorted keys, as for predecessor, successor, rank, range-count,
range-sum, range-minimum, range-maximum, and bounded-interval emptiness. The bridge
to a read floor in the second case is a per-level direct sum, supplied for
predecessor and successor by \Cref{lem:anchorzones} and for the rest by the two
lemmas below.

Call $Q$ \emph{additive} if $Q(S,q)=\sum_k\varphi(L_k,q)$ in an abelian group,
each contribution $\varphi(B,q)$ a function of $q$'s position in $B$; rank,
range-count, and unit-weight range-sum qualify. Let $g_Q(m)$ be the worst-case
reads to compute $\varphi(B,q)$ on one static block of $m$ samples under a hard
distribution $\mathcal D_m$, and call the family $\{\mathcal D_m\}$
\emph{embeddable} when every active size vector is realized by independent private
blocks and one shared query whose marginal at level $i$ is $\mathcal D_{m_i}$.

\begin{lemma}[Additive Direct Sum]\label{lem:additivesum}
For an additive $Q$ with an embeddable hard family of per-level complexity $g_Q$
under an oblivious schedule, the hard realization at time $t$ forces
$\Rhat_Q\ge\Omega\bigl(\sum_k g_Q(m_k)\bigr)$, deterministically and, with $g_Q$
the randomized static complexity, against bounded-error evaluators.
\end{lemma}

\begin{proof}
Draw the embedded product instance and fix a level $i$. Run the evaluator from
$B_i$'s draw alone, charging its reads inside $L_i$ and answering each read outside
$L_i$ for free from a block drawn on independent coins from its law given $q$,
which strong materialization and no-cascade make independent of $B_i$. The
evaluator returns $Q(S,q)$, and subtracting the $\sum_{j\ne i}\varphi(L_j,q)$ just
sampled leaves $\varphi(L_i,q)$ by the group inverse. This is a single-block
evaluator for $\mathcal D_{m_i}$ whose outside answers are advice independent of
$B_i$, so it cannot read $L_i$ below the static floor $g_Q(m_i)$. The blocks are
disjoint, so the reads sum.
\end{proof}

For abstract keys $g_Q(m)=\Theta(\log(m+1))$: computing $q$'s position in a block
is the comparison search of \Cref{lem:zerogap}, and a balanced tree matches it.
The lemma then gives $\Rhat_Q\ge c\sum_k\log_2(m_k+1)$ for rank, range-count, and
unit-weight range-sum, the floor predecessor pays; \Cref{sec:lifting} reuses it
with the integer and two-dimensional complexities.

\begin{lemma}[Idempotent Direct Sum]\label{lem:idempsum}
Fix a public size vector of an oblivious schedule and a bounded-error evaluator.
For bounded-interval emptiness, range-minimum, and range-maximum, each with a
fixed empty-interval default, some strongly materialized state with these sizes admits an interval $I$, empty
in every level, on which the evaluator has $\Rhat_Q\ge c\sum_i\log_2(m_i+1)$, for
an absolute $c>0$.
\end{lemma}

\begin{proof}
Give each level $i$ a low slab below $I$, a high slab above $I$, and a hot
subinterval inside $I$, the zones disjoint and ordered publicly by level index.
Draw $J_i$ uniform in $\{0,\dots,m_i\}$, placing the first $J_i$ keys of level $i$
in its low slab and the rest in its high slab, so $I$ holds no key and $Q$ returns
its default; the hot alternative, when $J_i<m_i$, moves $k_{i,J_i+1}$ into the hot
subinterval, making $I$ nonempty and changing the answer while every other
encoding stays fixed. Expose every level but $i$ in its empty state for free;
deciding the global answer is then deciding whether level $i$ leaves $I$ empty,
the task of \Cref{lem:zerogap} with $m=m_i$ and the endpoints of $I$ as the two
handles, at $\Omega(\log(m_i+1))$ reads. By no-cascade, reads elsewhere do not help,
the per-level costs sum, and some all-empty realization meets the bound.
\end{proof}

\begin{theorem}[Order/Range Dichotomy]\label{thm:dichotomy}
A decomposable query has $E\Rhat_Q=\Theta(n\log^2 n)$ when an $O(1)$-read per-level
certificate settles it and some unread live level can change its answer, as with
membership, point-emptiness, and the global extrema. When instead a level's
contribution turns on $q$'s position among its keys, as for predecessor and
successor, the additive rank, range-count, and range-sum, and the idempotent
range-minimum, range-maximum, and bounded-interval emptiness, the cost is
$E\Rhat_Q=\Theta(n\log^3 n)$, with no hypothesis on frontier or budget.
\end{theorem}

\begin{proof}
On the cheap side an $O(1)$-read witness per level answers $Q$ in
$\Rhat=O(W)=O(\log n)$, hence $E\Rhat=O(n\log^2 n)$, matched below by the
membership floor of \Cref{sec:membership}, whose argument needs only that an unread
live level can flip the answer. On the order-localizing side the direct sums give
$\Rhat_Q\ge c\sum_k\log_2(m_k+1)\ge c\max\{\What,\Phi(t)\}$ on the hard skeleton,
predecessor's own floor: \Cref{lem:anchorzones} for predecessor and, reversing the
order, successor; \Cref{lem:additivesum} for rank, range-count, and unit-weight
range-sum; and \Cref{lem:idempsum} for the idempotent range queries, where an
empty query interval keeps every level pivotal. The three-case assembly of
\Cref{thm:predecessor} then gives $\Omega(n\log^3 n)$, attained by binary
Bentley--Saxe carrying the matching static structure, a balanced tree with subtree
counts, prefix sums, or subtree extrema, at $O(\log m_k)$ reads per level.
\end{proof}

The line is hashing against within-level order localization: membership and
point-emptiness fall to a hash and the global extrema to one stored value per
level, all staying at the membership scale, while predecessor, successor, rank,
and the range queries turn on $q$'s position in each block and pay the logarithm.
Total count is cheaper still, fixed by the public size sequence at $\Theta(n)$ and
lying outside the trichotomy. What the logarithm buys is exactness of the order. The idempotent range queries
carry no group inverse, yet an empty interval makes
every level pivotal, a single hidden key overturning the empty default
(\Cref{lem:idempsum}). The logarithm is specific to abstract keys, and integer
keys shrink it (\Cref{sec:access}).

\begin{remark}
Approximation refunds the logarithm. A constant additive error $\pm\varepsilon n$
in rank or range-count, or a constant-precision quantile, separates only $O(1)$
buckets and returns the query to $\Theta(n\log^2 n)$; a multiplicative error keeps
the $\Theta(\log n)$ scales and does not.
\end{remark}

%% file: sections/select.tex
\section{Select}\label{sec:select}

Select proves the third line of \Cref{thm:main}.  Rank is a key-to-count query: the
query names a key, and each component contributes a count.  Select is the inverse,
count-to-key query: the query names only a rank $K$.  In a one-way scan the key
that will become the answer is not known while the early components are being
read, and those components cannot be revisited once later components reveal the
threshold.  The proof first turns this target-free difficulty into a cut lower
bound, and then shows that any structure with small read cost must repeatedly
rebuild a large component.

\subsection{The Cut Lower Bound}

Comparisons among keys already read are free, so the lower bound is on token reads
and the adversary attacks the unread keys. Two lemmas carry it: a bound on what a
single read can learn, and a cut floor built from it.

\begin{lemma}[Layout Channel]\label{lem:layoutchannel}
Fix the query $q$, its coins $\rho$, and the data-independent initial
configuration $\sigma_0$. Condition on the public layout $G$, the size sequence
together with the boundaries and addresses it determines. Let $X$ be the
stored keys and values. For one run of the forward head let $T$ be its number of
charged reads and $\Pi$ its transcript, the words it reads together with its
working memory. Then
\[
  I(X;\Pi\mid G,q,\rho,\sigma_0)\ \le\ b\,\mathbb E[T\mid G,q,\rho,\sigma_0]
   \ =\ O\bigl(\mathbb E[T]\log n\bigr),\qquad b=\Theta(\log n).
\]
\end{lemma}

\begin{proof}
Assume $\mathbb E[T]<\infty$, the bound being vacuous otherwise, and condition
throughout on $G,q,\rho,\sigma_0$. Record the $k$-th read as $V_k=(A_k,Y_k)$ when
$T\ge k$ and $V_k=\bot$ otherwise, where $A_k$ is the address read and $Y_k$ the
returned word together with the rank of its handle among those already read. A
query reads at most the $O(n)$ live tokens, so that rank fits in $O(\log n)$ bits
and $Y_k$ carries $b=\Theta(\log n)$ bits. The head is deterministic given the
conditioning and the reads so far, so $\{T\ge k\}$ is a function of
$V_1,\dots,V_{k-1}$, and on that event so is $A_k$, by \Cref{def:strongmat}; hence
$I(X;A_k\mid V_{<k})=0$ and
\[
  I(X;V_k\mid V_{<k})\ \le\ H(Y_k\mid V_{<k},A_k)\ \le\ b .
\]
Since $V_k=\bot$ when $T<k$, chain rule over $k\le K$ gives
$I(X;V_1,\dots,V_K)\le b\,\mathbb E[\min\{T,K\}]$, and $K\to\infty$ gives
$I(X;V^\infty)\le b\,\mathbb E[T]$. The transcript $\Pi$ is a function of the
conditioning and $V^\infty$ by the no-external-state clause, so data processing
closes it.
\end{proof}

For a randomized evaluator the coins $\rho$ are independent of the data, so the
bound holds with $\Pi$ replaced by $(\Pi,\rho)$ and averaged over $\rho$.

\begin{lemma}[Cut Floor]\label{lem:cutfloor}
At a boundary between two active levels holding $P$ samples before it and $U$
after, with $r=\min\{P,U\}$, some input family with the fixed rank $K=r$ forces
every bounded-error forward evaluator to read $\Omega(r)$ words before the head
crosses the boundary. Hence $\Rhat(t)\ge c\min\{P,U\}$ for an absolute $c>0$.
\end{lemma}

\begin{proof}
Place $r$ disjoint rank slabs $S_1<\cdots<S_r$ in the prefix, separated by fixed
guards, each slab a public range of $Q=\mathrm{poly}(n)$ candidate positions. Draw
$a_i$ uniformly and independently from $S_i$ and let $Z_i$ name which candidate it
is, so the $a_i$ carry $\Theta(\log n)$ independent bits while their order is fixed
by the slabs; put the remaining $P-r$ prefix keys in a guard above $S_r$. The
suffix encodes an index $j$: place $r-j$ of its keys in a guard below $S_1$ and the
rest in a guard above every slab. Then exactly $r-j$ suffix keys and the $j-1$
slab keys $a_1,\dots,a_{j-1}$ lie below $a_j$, so $a_j$ is the $r$-th smallest of
the union, while the prefix and the rank $K=r$ are the same for every $j$.

Draw $J$ uniform on $\{1,\dots,r\}$, independent of $Z=(Z_1,\dots,Z_r)$, and let
$\Pi$ be the transcript and coins at the crossing. The prefix is read before the
crossing and does not depend on $J$, so $(Z,\Pi)$ is independent of $J$, while the
suffix determines $J$ and carries nothing about $Z$. A correct answer is $a_J$,
which the head must decode from $\Pi$ and $J$ since it cannot return to the prefix;
so by Fano $I(Z_j;\Pi\mid J=j)=\Omega(\log Q)$, equal to $I(Z_j;\Pi)$ by that
independence. As the $Z_i$ are independent,
$I(Z;\Pi)\ge\sum_i I(Z_i;\Pi)=\Omega(r\log n)$. Since $Z$ is a function of $X$ and
$\Pi$ records the $T$ prefix reads, \Cref{lem:layoutchannel} gives
$\Omega(r\log n)\le I(X;\Pi)\le O(\mathbb E[T]\log n)$, whence $\mathbb
E[T]=\Omega(r)$ and $\Rhat(t)\ge\mathbb E[T]=\Omega(r)$.
\end{proof}

The guards keep the suffix clear of $Z$: its keys sit in fixed slabs, never at a
data-dependent value such as $a_1-1$, so reading them reveals only $j$. The floor
rests on a single boundary rather than a sum over levels, but that is enough: a
low read cost forces almost all the live mass into one level, and building such a
level again and again is the write cost select cannot escape.

\subsection{The Quadratic Product Bound}

\begin{theorem}[Quadratic Select Bound]\label{thm:select}
Strongly materialized select has $\min E\Rhat=\Theta(n^2)$, with no hypothesis on
the frontier or the budget, and the lower bound holds against bounded-error query
randomization.
\end{theorem}

\begin{proof}
Write $R=\Rhat$ and $s=R/c$ with $c$ the constant of \Cref{lem:cutfloor}. The live
mass at time $t$ is $t$, and at every boundary $\min\{P,U\}\le R/c=s$, so each
boundary mass is $\le s$ or $\ge t-s$ and the level spanning $(s,t-s)$ has size
$\ge t-2s$. Suppose $R<cn/8$, so $s<n/8$ and this level exceeds $n/4$ throughout
$t\in[n/2,n]$. Take times $t_0<\cdots<t_K$ spaced by $\Delta=\lceil 2s\rceil+1$
across $[n/2,n]$, so $K=\Omega(n/(R+1))$, and let $B_i$ be such a level at $t_i$.
Being materialized, $B_i$ is static after birth, so $|B_i|\le t_i$, while
$|B_{i+1}|\ge t_{i+1}-2s>t_i\ge|B_i|$; the $B_i$ are therefore distinct components,
each costing $\Omega(n)$ tokens. Hence $E=\Omega(n^2/(R+1))$, and as $R\ge1$,
$E\Rhat=\Omega(n^2)$. If instead $R\ge cn/8$, then $E\ge n$ gives
$E\Rhat\ge nR=\Omega(n^2)$.

For the upper bound, keep one large level together with a suffix of at most $s$
singletons, merging the whole stack into a fresh large level every $s$ inserts.
Rebuilding the large level at sizes $s,2s,\dots,n$ costs $E=\Theta(n^2/s)$. The large
level is a sorted array at the bottom of the stack, so the head jumps by a computed
offset to its rank window $[K-s,K]$, reads those $O(s)$ keys and the $\le s$
singletons above it, and selects the global $K$-th among these candidates; the
suffix shifts the answer's rank in the large level by at most $s$, so the window
holds it, and $\Rhat=O(s)$. The product is $\Theta(n^2)$ at every $s$.
\end{proof}

The three query classes now stand in a line, governed by what the query hands the
forward head:
\[
  \underbrace{\Theta(n\log^2 n)}_{\text{hashable witness}}\ <\
  \underbrace{\Theta(n\log^3 n)}_{\text{known key value}}\ \ll\
  \underbrace{\Theta(n^2)}_{\text{rank, no value}} .
\]
The first gap is one logarithm, order against hashing; the second is a polynomial. Rank and select are inverse operations, value to count and
count to value, and select's answer is itself only $O(\log n)$ bits, yet the two
lie polynomially apart: the gap is made by the access discipline, not by the
information in the answer. Select gives the head no value to aim at, and that is
what carries it past predecessor. The cost is a property of the forward head, not of
the key model, since the lower bound forces its reads by completing cells the head
has not visited and never inspects a key's representation; it therefore persists
for integer keys, where the word-RAM that narrows the order gap (\Cref{sec:access})
still leaves select at $\Theta(n^2)$.

\begin{remark}
Constant-precision approximate select returns to $\Theta(n\log^2 n)$: $O(1)$
samples at evenly spaced ranks per level form a forward quantile sketch that needs
no refinement of the exact threshold. Exactness is again what the cost buys.
\end{remark}

\begin{remark}
A cross-level rank catalog would answer select in $O(1)$ reads, but maintaining one
under LIFO appears to cost $\Omega(n^2)$ edits, which would carry the $\Theta(n^2)$
past strong materialization, as \Cref{sec:nonmat} carries the predecessor bound. We
do not pursue this here.
\end{remark}

%% file: sections/lifting.tex
\section{Static-to-Dynamic Lifting}\label{sec:lifting}

Every query that pays the order logarithm lands at the same $n\log^3 n$, and a
single reduction explains why. Each pays, per level, the static cost of locating
the query there; the model forbids sharing those costs across levels, so the read
cost is their sum. For an additive query this turns the dynamization optimum into
the edit budget times a sum of static per-level complexities, a dynamic lower bound
read off a static one. The logarithm of order is then no artifact of maintenance
but the faithful lift of a static search cost.

Recall the additive queries of \Cref{sec:order}: $Q(S,q)=\sum_k\varphi(L_k,q)$ in
an abelian group, with per-level static complexity $g_Q(m)$, the reads to compute
one block's contribution $\varphi(B,q)$, and embeddable hard families that plant an
independent
instance in every level (\Cref{lem:additivesum}).

\begin{theorem}[Static-to-Dynamic Lifting]\label{thm:lifting}
Let $Q$ be additive with an embeddable hard family and per-level static complexity
$g_Q$: a forward-readable static structure attains $O(g_Q(m))$ reads on a block of
$m$ samples, and the hard family forces $\Omega(g_Q(m))$. Suppose $g_Q$ is
nondecreasing with $g_Q(1)\ge1$ and of linear rate $g_Q(2^k)\ge\zeta_Q k$ for some
$\zeta_Q\in(0,1]$, and $\Sigma_Q:=\sum_{k=0}^{L}g_Q(2^k)=\Theta(\zeta_Q L^2)$. Then
every strongly materialized algorithm with a bounded-error evaluator has
$E\Rhat_Q=\Omega(nL\,\Sigma_Q)$, which binary Bentley--Saxe attains, so
$\min E\Rhat_Q=\Theta(nL\,\Sigma_Q)$.
\end{theorem}

\begin{proof}
For the upper bound, the schedule of \Cref{prop:baseline} carries that static
structure at each level, a balanced search tree with subtree aggregates for the
additive queries. Then $E=O(nL)$, and the active sizes being distinct powers of
two, $\Rhat=\sum_{\text{active}}g_Q(m_k)\le\Sigma_Q$.

For the lower bound, \Cref{lem:additivesum} gives
$\Rhat_Q(t)=\Omega(\sum_k g_Q(m_k))$; the linear rate and the floor $g_Q(1)\ge1$
make $g_Q(m)=\Omega(\zeta_Q\log_2(m+1))$, so
$\Rhat_Q(t)=\Omega(\zeta_Q\sum_k\log_2(m_k+1))$. Since $\sum_k\log_2(m_k+1)$ is at
least $\max\{W(t),\Phi(t)\}$ at all times and at least $\log_2(n+1)$ at the final
skeleton, this is the read floor of \Cref{thm:predecessor} scaled by $\zeta_Q$. Its
three-case assembly then applies: the wide-frontier and deep cases exceed
$nL\,\Sigma_Q$, while the size-diversity case multiplies $E=\Omega(n\Dtil)$ by
$\Rhat_Q=\Omega(\zeta_Q\Phi(t^\star))$ and closes with the calculus floor, giving
$E\Rhat_Q=\Omega(\zeta_Q nL^3)$. By $\Sigma_Q=\Theta(\zeta_Q L^2)$ this is
$\Omega(nL\,\Sigma_Q)$.
\end{proof}

The middle row of the dichotomy is now one line: abstract-key rank, range-count,
and range-sum have $g_Q(m)=\Theta(\log(m+1))$, hence $\zeta_Q=\Theta(1)$,
$\Sigma_Q=\Theta(L^2)$, and $E\Rhat_Q=\Theta(nL^3)$. Beyond recovering the additive
row of \Cref{thm:dichotomy}, the lift names the structure of the cost. The read
bound is a sum of per-level static search complexities, one term per live level,
because maintenance does not inflate a single level and cascading cannot amortize
across them. That this sum cannot stay below $L^2$ is the size-diversity of
\Cref{sec:order}: that section forces the sum up to $L^2$, and this one explains why
the cost is that sum to begin with. The two are the lower-bound face and the
structural face of the same $n\log^3 n$.

\begin{remark}
Because $g_Q$ is a static complexity, the lift inherits its robustness. Integer keys
shrink the per-level cost and carry rank, range-count, and range-sum to
$\Theta(n\log^2 n\,\lambda)$, while in two dimensions the separation survives the
word-RAM, with range-counting lifting to $\Theta(n\log^3 n/\lambda)$ and no van Emde
Boas collapse. Both are taken up in \Cref{sec:access}; the integer bounds there are
deterministic and Las Vegas cell-probe bounds, and a bounded-error version would
need a randomized static one. Integer predecessor, not an additive
query, stays open between $1$ and $\lambda$.
\end{remark}

%% file: sections/accessmodel.tex
\section{Access-Model Variants}\label{sec:access}

The classification of \Cref{thm:main} is stated for abstract keys, strong
materialization, and one-way scans.  This section records what happens when one
restriction at a time is relaxed.  These variants are not needed for the main
lower bounds; their role is to identify which assumption supports which
separation.  Integer keys reduce the cost of searching inside one component,
random access removes the no-backseek obstruction for select, and
non-materialized cascade data can remove the abstract-key predecessor gap only
when it can be placed in the right scan direction.

\subsection{Integer Keys}

The order logarithm is a property of abstract keys. An integer key in a
$\mathrm{poly}(n)$ universe is searchable by the word-RAM: a y-fast
trie~\cite{willard1983} finds a run's predecessor in $O(\lambda)$ reads and a
fusion tree~\cite{fredmanwillard1993} in $O(\log_w m)$, both forward-readable under
the round-by-round layout of \Cref{prop:baseline}. Predecessor then costs
$O(n\log^2 n\,\lambda)$, its gap over membership shrinking from $\log n$ to
$\lambda$, the way radix sort undercuts the comparison-sorting bound without
refuting it.

The collapse is not uniform, and the reason is what makes abstract keys hard. An
additive query reads a contribution from every run, so each must still be searched,
now at integer cost: a rank among $m$ integers is a predecessor search of static
complexity $g_Q(m)=\Theta(\max\{1,\min\{\log_w m,\lambda\}\})$, the y-fast or fusion layout
above, matched from below by P\u{a}tra\c{s}cu and
Thorup~\cite{patrascuthorup2006,patrascuthorup2007}, deterministic and Las Vegas
alike. Its binary sum is $\Sigma_Q=\Theta(L\lambda)$, so \Cref{thm:lifting} places
rank, range-count, and range-sum at $\Theta(n\log^2 n\,\lambda)$, exactly
$\Theta(\lambda)$ above membership.

Predecessor behaves differently, because it is settled by one run and the rest need
only be ruled out. Ruling run $i$ out asks whether its keys meet the open interval
between the running best and the query, a range-emptiness test the word-RAM answers
in $O(1)$~\cite{abr2001} but abstract keys answer only by localizing the query
inside the run, the $\Omega(\log m)$ step behind the separation. The per-run floor
the additive queries enjoy has no integer counterpart here, so the exact integer
cost of predecessor stays open between $\Theta(n\log^2 n)$ and
$\Theta(n\log^2 n\,\lambda)$: a lower bound would need a per-run test that stays
$\Omega(\lambda)$-hard under a shared query, and the natural one, whether a run
holds the global predecessor, is the range-emptiness integers trivialize. The key
models part exactly here.

The collapse is one-dimensional. Two-dimensional range-counting has no van Emde
Boas speedup, its cell-probe complexity $\Omega(\log m/\lambda)$ at linear
space~\cite{patrascu2007} matched by a range tree~\cite{chanwilkinson2013}, so
lifting $g_Q(m)=\Theta(\max\{1,\log m/\lambda\})$ gives $\Theta(n\log^3 n/\lambda)$, above
membership by $\Theta(\log n/\log\log n)$. The separation is no artifact of
unstructured keys; it survives the word-RAM as soon as the queries are
two-dimensional, the regime of multi-attribute range search.

\subsection{Random Access}

Select's polynomial cost is the signature of the no-backseek head, and it vanishes
the moment the head may return. The obstacle of \Cref{sec:select} was an answer
threshold fixed by all runs jointly, unknown while the early runs are read and out
of reach once they are passed. Random access dissolves it: selecting the $K$-th
smallest among $W$ sorted runs is classical, $O(W+\sum_i\log(k_i+1))$ comparisons
with $k_i$ run $i$'s share of the $K$ smallest~\cite{fredericksonjohnson1984,kkzz2019},
which is $O(\log^2 n)$ at a diverse frontier. Random access brings select to
$O(n\log^3 n)$, the value-query scale, with the exact optimum open.

It does no more. Predecessor stays at $n\log^3 n$ under random access, since every
run must still be searched and no run bridges another; what random access removes
is not the per-run search but the forward head's inability to revisit an internal
threshold. The polynomial gap was never in the answer, only $O(\log n)$ bits, but
in that head. Integer keys touch only the order gap and random access only select's,
each isolating the restriction its cost depends on.

\subsection{Deque Edits}

Two restrictions guard the cascade together: a merge stack writes only at the tail,
behind the forward head, and its runs are strongly materialized, so no run holds
order information about another. Relax both, writing at the head's entrance and
letting a young run carry a bridge into an older one, and the separation falls.
Model the relaxation by \emph{deque} edits, popping and pushing at both ends of the
token string while keeping the same forward, no-backseek head, so a component may be
placed at position $0$, ahead of all the head will read; dropping materialization
lets a young component store forward handles into the older ones, charged to itself
and leaving them unedited.

\begin{theorem}[Deque Collapse]\label{thm:dequecollapse}
Under deque edits with non-materialized runs, abstract-key predecessor admits
$E=O(n\log n)$ and $\Rhat=O(\log n)$, so $E\Rhat=O(n\log^2 n)$ and the order
separation does not survive.
\end{theorem}

\begin{proof}
A useful bridge from $U$ to $V$ needs $U$ read before $V$ and $V$ born before $U$
(\Cref{lem:causalbridge}), and tail-only editing forces birth order to follow read
order and kills it. Deque editing breaks the coupling: lay the active components
youngest first, $C_1,\dots,C_W$ read from the entrance with sizes increasing, so
$C_i$ is read before $C_{i+1}$ while $C_{i+1}$ is the older of the two, the very
order a cascade needs.

Run the Bentley--Saxe counter mirrored at the entrance: each insert pushes a
singleton at the front and, while the two front components share a size, pops both,
merges, and pushes one new front component. A component of $m_i$ keys stores its
block as a forward-readable balanced tree together with a sparse cascade into the
next-older $C_{i+1}$, cutting $C_{i+1}$'s order into $O(m_i)$ intervals of span
$O(m_{i+1}/m_i)$ and keeping, per gap, its own predecessor and a forward handle to
the precomputed subtree root for that interval, a node already inside $C_{i+1}$'s
own tree once $C_{i+1}$ is serialized. Each older component is serialized once with
its roots in place, so a younger one adds only its $O(m_i)$ handles and never edits
it; hence $|C_i|=O(m_i)$, and a size-$2^k$ component is built $O(n/2^k)$ times at
$O(2^k)$ tokens, giving $E=O(n\log n)$.

A query searches $C_1$ in $O(\log m_1)$ reads, then at each step follows the stored
handle into a span-$O(1+m_{i+1}/m_i)$ interval of $C_{i+1}$, searches its
precomputed subtree, updates the best candidate, and moves on without a backseek.
The reads telescope,
\[
  \Rhat=O\Bigl(W+\log m_1+\sum_{i\ge2}\log\tfrac{m_i}{m_{i-1}}\Bigr)
   =O(W+\log m_W)=O(\log n).\qedhere
\]
\end{proof}

The collapse is one of orientation, not of two-endedness. A queue, pushing at the
tail and popping at the head, still builds its components at the scan's exit and
does not collapse; what predecessor pays for under tail-only editing is exactly
that new connectivity must be written behind the forward scan. The gap sits at the
missing cascade, and the cascade the deque builds is non-materialized, a young run
pointing into an older one, so the collapse leaves the materialized model and the
tail-only one together. This is why random access alone, the runs still
materialized, does not collapse predecessor. The matching $\Omega(n\log^2 n)$ would
follow from a membership floor under deque edits, a minor open point since the
universal floor of \Cref{sec:membership} uses the tail-only frontier-width bound; the
collapse out of the $n\log^3 n$ class is unconditional. Afshani's dynamic
fractional-cascading lower bound~\cite{afshani2021} forces either fully dynamic
updates or worst-case update time, while the construction here is insertion-only
with amortized rebuilding, and escapes both.

\subsection{LSM Read Amplification}

The read gap reaches past edit histories to the runs in place at query time, though
only under a hypothesis. A run can hold its median and successor in an $O(1)$ header
and certify a query's predecessor between them in two reads, so overlap by itself
does not force $\Omega(\log m)$. What forces it is a query distribution that, after
all outside information is exposed, still leaves $\Omega(\log m_i)$ of conditional
gap entropy in each overlapped run, exactly the hard family of \Cref{lem:anchorzones}.

\begin{proposition}[Range Read Amplification]\label{prop:lsmreadamp}
Let $R_1,\dots,R_W$ be active sorted runs, pairwise no-cascade and carrying no
cross-run order catalog. Under a query distribution leaving $\Omega(\log m_i)$
conditional gap entropy in each run it crosses, an exact range or order query has
read cost $\Omega(\sum_{i\,\text{crossed}}\log m_i)$, while a point query is settled
in $O(1)$ reads per run.
\end{proposition}

\begin{proof}
The runs being mutually no-cascade, a read in one cannot shrink another's candidate
set, so the per-run argument of \Cref{lem:anchorzones} applies to each crossed run
on its own: the hypothesis leaves $\Omega(\log m_i)$ gap entropy there, and
localizing the endpoint costs $\Omega(\log m_i)$ reads. Carrying no cross-run
catalog, the runs do not share the cost, so it sums. A point query is hashable,
settled per run by a filter in $O(1)$.
\end{proof}

The bound is on the runs in place, however built or read, so it holds
in any leveled sorted-run store. In a random-access store the per-run filters let a
point query skip the runs that miss it, leaving $O(1)$; a range or order query has
no such escape and localizes in every run it spans,
$\Theta(\log_T n)$ at size-ratio $T$, each localization irreducible by
\Cref{prop:lsmreadamp}. This
range-over-point read amplification closes only two ways: a word-RAM integer index,
which serves one dimension but not string keys or multi-attribute ranges, where the
two-dimensional separation persists; or a cross-run order catalog, which appears to
need $\Omega(n^2)$ edits to maintain under append-only merging (\Cref{sec:select}).
For write-optimized, string-keyed, multi-attribute engines in the insertion-only
regime the model captures, the range read amplification is intrinsic, not an
implementation artifact~\cite{oneil1996,monkey2017,dostoevsky2018}.

\subsection{Intermediate Queries}

Order at $n\log^3 n$ and select at $n^2$ are the two ends of one scale, set by how
much the query reveals about its answer's value. A $w$-windowed select, the $K$-th
smallest under the promise that the answer lies in a bracket of at most $w$ live
keys, interpolates between them: at $w=O(1)$ the bracket pins the value and the
query is an order query, at $w=N$ the promise is empty and it is select. The cut
floor restricted to the bracket gives $\Rhat\ge c\min\{w,P,U\}$, only the $\le w$
keys inside it carrying the ambiguity. The matching write tradeoff across the
interior is open, and we leave it.

%% file: sections/nonmaterialized.tex
\section{Cascades Beyond Materialization}\label{sec:nonmat}

Strong materialization (\Cref{def:strongmat}) makes each component a self-contained
sorted file, the locality the lower bounds above turn on: a read in one component
carries no order information about another. This section removes exactly that. A
young component may now keep a directory, copied separators, or a search summary for
older keys. The rest of the serialized model stays: LIFO edits, exact layout, no
sample-dependent state outside the token string, oracle abstract keys, and a forward
no-backseek head.

Here the separation turns genuinely dynamic. The earlier proof priced the write
budget against the frontier width and the live search potential, but a
non-materialized suffix can fuse many old searches into one directory lookup, so the
frontier is no longer the right primitive. What survives is sharper: while a token
sits in the longest common prefix of two consecutive states, the transcript by which
a query reaches it is fixed. This yields a canonical alphabetic search tree laid out
in stack order, and the reduction it supports proves the $n\log^3 n$ bound at
near-optimal read cost and isolates the rest as a single suffix-recourse conjecture.

Write $s_t$ for the token string after update $t$. A LIFO edit has the form
\[
  s_t=P_tA_t,\qquad s_{t+1}=P_tB_t,
\]
with $P_t$ the longest common prefix and charged cost $|A_t|+|B_t|$; set
$R=\lceil\Rhat\rceil$.

\begin{lemma}[Prefix-Transcript Invariance]\label{lem:nonmat-prefix}
Fix a deterministic evaluator and a query $q$. On the states $s=PA$ and $s'=PB$, the
two executions agree until the head first leaves $P$.
\end{lemma}

\begin{proof}
Both start in the same control state, query, and head position. While the head stays
in $P$, each step reads the same token at the same position, exact layout included,
runs the same computation, sees the same comparison and equality outcomes, moves to
the same control state, and computes the same forward skip, so the executions agree
inductively. A token of the new suffix $B$ is not reached before the head leaves $P$,
so it cannot alter any earlier branch.
\end{proof}

This is stronger than the causal-bridge obstruction of \Cref{lem:causalbridge}: not
only can a later directory not point backward in time, a later suffix cannot change
the transcript by which a query reaches any retained token.

\begin{lemma}[Alphabetic-Tree Normal Form]\label{lem:nonmat-tree}
Fix a state of $N$ distinct abstract keys with distinct values and a deterministic
exact-predecessor evaluator of read cost at most $R$. Placing a single nonmember
handle in each of the $N+1$ open gaps in turn, the evaluator induces an alphabetic
comparison tree of $N+1$ leaves and height at most $R$, each internal separator
witnessed by a token occurrence that reveals the corresponding stored key, with the
witness positions
strictly increasing along every root-to-leaf path.
\end{lemma}

\begin{proof}
Let the $N+1$ alternatives be one nonmember handle $\xi$ with a fixed equality label,
distinct from every stored one, that the oracle drops into each gap in turn; across
the alternatives $\xi$'s equality tests stay false and its hash address stays fixed,
so only a comparison against a revealed stored key can branch by gap. Unfold the
evaluator into its decision tree and normalize it to compare $\xi$ with each handle
the moment that handle is read, which adds no charged reads and puts the
order-sensitive comparisons in read order. Order-independent reads then give the same
outcome in every gap and are quotiented out. Each gap's answer is the value of its
predecessor, distinct across the $N+1$ gaps, so their transcripts carry distinct
leaves. For the two gaps flanking a stored key $x$ the
predecessors differ; with $\xi$ placed close enough to $x$, every comparison against
a $y\ne x$ agrees on both sides and no hash or equality label reveals an order
position, so the first order-sensitive divergence is a comparison against a token
revealing $x$, taken as the canonical witness of the boundary at $x$. These
comparisons separate contiguous intervals of query order, so after suppressing unary
nodes they form an alphabetic tree of $N+1$ leaves; each charged read reveals at most
one stored handle, so a root-to-leaf path holds at most $R$ of them; and the handles
being compared in read order under a head that never moves back, the witness
positions strictly increase along each path.
\end{proof}

The oracle-key hypothesis carries this. A directory entry can say where to skip, and
a young component can copy an old handle, but a token that reveals no stored key
cannot act as an order-bearing separator: an old handle denotes the same abstract
key, its hash leaks no position, and by \Cref{lem:nonmat-prefix} the transcript
reaching its token is fixed, so a later directory cannot reinterpret it as a fresh
separator before it is read. This shuts the loophole by which metadata might
manufacture separators for free.

A retained canonical witness keeps its address in the tree. If its token lies in
$P_t$, so does every earlier witness on its query path, and \Cref{lem:nonmat-prefix}
fixes the incoming transcript and the ancestor chain; a suffix witness occurs later
in the string, so it may refine a retained leaf but cannot become an ancestor of a
retained witness.

\begin{theorem}[Labeling Reduction]\label{thm:nonmat-labeling}
Every deterministic non-materialized exact-predecessor structure with write cost $E$
and prefix-worst read cost $\Rhat$ induces an online order-preserving labeling of the
inserted keys, with label range at most $2^{R+1}$ and total relabeling at most $E$,
in which changing an old label means deleting that key's canonical witness from the
physical suffix and appending replacements.
\end{theorem}

\begin{proof}
At each time take the canonical tree of \Cref{lem:nonmat-tree} and give a separator
at depth $d$ with binary address $a$ the dyadic label $(2a+1)2^{R-d}$, the midpoint
of its depth-$d$ interval refined to depth $R$. These are strictly ordered by the
inorder of the separators and lie in $\{1,\dots,2^{R+1}-1\}$, a range for the
reduction and not a constraint on the keys. Under one edit $P_tA_t\to P_tB_t$, an old
key whose witness stays in $P_t$ keeps its depth, address, and label, so every old
key that is relabeled has its witness in the popped $A_t$; distinct keys have
distinct witnesses, so at most $|A_t|$ keys are relabeled, and the new key draws its
first label from $B_t$. The update costs at most $|A_t|+|B_t|$, and the sum is at
most $E$. The retained witnesses are a prefix of the physical witness order and the
replacements are written in the appended suffix, the recourse the statement names.
\end{proof}

\begin{proposition}[Near-Static Regime]\label{prop:nonmat-regimes}
A deterministic non-materialized exact-predecessor structure with
$\Rhat\le\log_2 n+O(1)$ has $E\Rhat=\Omega(n\log^3 n)$.
\end{proposition}

\begin{proof}
The labeling lower bound is adaptive over a totally ordered key set, so against a
deterministic structure the adversary fixes a sequence of $n$ keys with distinct
values whose order alone is read through the oracle, realized by distinct ranks in
the $\mathrm{poly}(n)$ universe. With $\Rhat\le L+O(1)$ the induced label range
$2^{R+1}=O(n)$ stays well inside that universe, so the adversary can refine a label
interval past exhaustion, and the linear-range labeling lower bound of Bul\'anek,
Kouck\'y, and Saks~\cite{bks2012} forces $\Omega(nL^2)$ relabelings, giving
$E=\Omega(nL^2)$. Exact predecessor search needs $\Rhat=\Omega(L)$, so
$E\Rhat=\Omega(nL^3)$.
\end{proof}

The reduction reaches no further.  The $n\log^3 n$ bound rests on the
linear-range labeling lower bound, hence on a label range $O(n)$ and
$\Rhat\le\log_2 n+O(1)$.  For larger read cost the range outgrows the linear
regime and, once it passes the $\mathrm{poly}(n)$ key universe, leaves ordinary
labeling with no relabelings to force.  What the reduction keeps and ordinary
labeling discards is sharper: an old label changes only through a physical suffix
rebuild of the alphabetic tree.  Extending the lower bound to all read costs would
amount to proving that suffix-recourse alphabetic trees must delete
$\Omega(nL^3/R)$ old separators over some insertion sequence when their height is
bounded by $R$.  We record this as the residual non-materialized kernel, not as a
claim needed by the main theorem.

\begin{remark}
A sparse directory genuinely escapes the materialized proof. For an old sorted set
$A$ and a younger batch $B$, the young suffix can store, for each nonempty gap $g$ of
$A$, a static predecessor structure over $B_g=\{x\in B:x\text{ falls in }g\}$, of
total size $O(|B|)$; having searched $A$ the evaluator knows $g$ and searches only
$B_g$, fusing many localization scales without copying $A$. What it cannot escape is
temporal renewal: once several batches enter one gap, a late directory names the
newest bucket but the forward head cannot turn back to an older body, so making the
newest resolver forward-readable forces a suffix rebuild. Whether this renewal can be
paid once for all nested scales is the residual question. Three views sharpen it. Witnesses
sit parent before descendant, so a popped witness pops its descendants and the
deletions are disjoint complete subtrees, not arbitrary relabelings. A direct-sum
potential suffices while a prefix skeleton survives, the local recourse potentials
adding over the external cells; what is missing is the renewal inequality once the
skeleton itself is popped, when one rebuilt subtree can refresh every nested resolver
beneath it at a stroke. And in rectangle terms each witness is a block laminar in
both time and query order, a laminarity broken by path copying, backward pointers, or
entrance writes, which is why random-access cascades, the cache-oblivious lookahead
array~\cite{cola2007}, and dynamic fractional cascading~\cite{afshani2021} sit beside
this problem rather than settle it.
\end{remark}

Materialization is the last assumption the model sheds, and the one this paper leaves
as a frontier rather than a crutch: the predecessor bound of \Cref{sec:order} needs
none of this section, the non-materialized separation is unconditional at
near-optimal read cost, and the residue is a suffix-recourse alphabetic-tree problem. The
section is deterministic. A randomized extension cannot follow from the labeling
reduction alone, since randomized online labeling beats the deterministic
$\log^2 n$ barrier in the linear range~\cite{bcfc2022}; any randomized proof would
have to use the suffix-renewal constraint itself, not label density.